\documentclass[12pt]{article}
\usepackage{amsfonts,amsmath,amssymb,amsthm}
\usepackage{comment}
\usepackage[text={5.5in,10in},centering]{geometry}
\usepackage{hyperref}

\newtheorem{theorem}{Theorem}[section]

\newtheorem{definition}[theorem]{Definition}
\newtheorem{example}[theorem]{Example}
\newtheorem{lemma}[theorem]{Lemma}
\newtheorem{postulate}{Postulate}
\newtheorem{remark}[theorem]{Remark}

\renewcommand\a{\ensuremath{\bar a}}

\newcommand{\cI}{{\cal I}}
\newcommand{\cS}{{\cal S}}
\newcommand\D{\ensuremath{\Delta}}

\newcommand\false{\ensuremath{\textsf{false}}}

\newcommand\s{\ensuremath{\sigma}}
\newcommand\set[1]{\ensuremath{\{#1\}}}
\renewcommand{\phi}{\varphi}
\newcommand\qef{\hfill$\triangleleft$} 
\newcommand\qefhere{\tag*{$\triangleleft$}} 
\newcommand\true{\ensuremath{\textsf{true}}}
\newcommand\U{\ensuremath{\Upsilon}}
\newcommand\undef{\ensuremath{\textsf{undef}}}
\newcommand\V{\ensuremath{\mathcal V}}

\title{A more abstract bounded exploration postulate}
\author{Yuri Gurevich and Tatiana Yavorskaya}
\date{}
\begin{document}
\maketitle

\begin{abstract}
In article ``Sequential abstract state machines capture
sequential algorithms'', one of us axiomatized sequential algorithms by means of three postulates: sequential time, abstract state, and bounded exploration postulates.
Here we give a more abstract version of the bounded exploration postulate which is closer in spirit to the abstract state postulate.
In the presence of the sequential time and abstract state postulates, our postulate is equivalent to the origingal bounded exploration postulate.
\end{abstract}

\section{Introduction}

This paper is essentially an oversized footnote to \cite{G141} where sequential algorithms are axiomatized by means of three postulates: sequential time, abstract state, and bounded exploration postulates.
To make our exposition more self-contained, we restate the postulates.

\begin{definition}[Sequential Algorithms]\rm
A \emph{sequential algorithm} $A$ is defined by means of the sequential time, abstract state, and bounded exploration postulates below. \qef
\end{definition}

\begin{postulate}[Sequential Time]\rm
$A$ is associated with a nonempty set $\cS(A)$ (of {\em states}), a nonempty subset $\cI(A)$ (of {\em initial states}), and a map $\tau_A :\cS(A)\to\cS(A)$ (the {\em one-step transformation}). \qef
\end{postulate}

We write simply $\cS, \cI$, and $\tau$ when $A$ is clear from the context.
The original version of the postulate in \cite{G141} did not require that $\cS$ and $\cI$ be nonempty.  This reasonable modification is due to \cite{G166}.

\begin{postulate}[Abstract State]\mbox{}\rm
\begin{itemize}
\item
The states are first-order structures of the same vocabulary \U\ (or $\U(A)$, the \emph{vocabulary} of $A$), and $\cS, \cI$ are closed under isomorphisms.
\item
The one-step transformation $\tau$ does not change the base set of any state, and any isomorphism from a state $X$ to a state $Y$ is also an isomorphism from $\tau(X)$ to $\tau(Y)$. \qef
\end{itemize}
\end{postulate}

Symbols in \U\ are function symbols; relation symbols are viewed as function symbols whose interpretations taking Boolean values.
Each function symbol $f$ has some number $j$ of argument places; $j$ is the \emph{arity} of $f$. The arity may be zero; in the logic literature nullary functions symbols are often called constants.
Terms are built as usual from nullary symbols but means of symbols of positive arity.

The vocabulary \U\ contains \emph{logical symbols} \true, \false, \undef, the equality sign, and the standard propositional connectives; the other symbols in \U\ are \emph{nonlogical}.
In any (first-order) \U\ structure, the values of \true, \false, and \undef\ are distinct \emph{logical elements};
the other elements are \emph{nonlogical}.

Let $X$ be a state, $f$ range over \U, $j$ be the arity of $f$, and $x_0, x_1, \dots, x_j$ range over the base set $|X|$ of $X$.
A triple $(f,(x_1,\dots,x_j),x_0)$ is an \emph{update} of $X$ if $x_0 = f(x_1,\dots,x_ja)$ in $\tau(X)$. An update $(f,\a,b)$ is \emph{nontrivial} if $x_0 \ne f(x_1,\dots,x_j)$ in $X$.
The \emph{update set} $\D(X)$ of the algorithm $A$ at state $X$ is the set of nontrivial updates of $X$.

States $X$ and $Y$ {\em coincide} over a set $T$ of  \U\ terms, symbolically $X \overset T= Y$, if every $t\in T$ has the same value in $X$ and $Y$.

\begin{postulate}[Bounded Exploration]
There exists a finite set $T$ of closed \U\ terms (called \emph{a bounded exploration witness}) such that
\[ X\overset T= Y \implies \D(X) = \D(Y)\quad
\text{for all states }X,Y. \qefhere \]
\end{postulate}

If $T\subseteq T'$ are sets of terms and $T$ is a bounded exploration witness, then $T'$ is also a bounded exploration witness.
For example, $T'$ could comprise $T$ and all subterms of terms in $T$, so that $T'$ is closed under subterms.

This completes the definition of sequential algorithms.

\bigskip
The bounded exploration postulate arguably contradicts the spirit of the abstract state postulate according to which a state is just
a presentation of its isomorphism type so that only the isomorphism type of the state is important.  In the bounded exploration postulate above, it is essential that the states $X$ and $Y$ are concrete.

In technical report \cite{G177} we gave a more abstract form of the bounded exploration postulate that is in the spirit of the abstract state postulate and that is equivalent to the original bounded exploration postulate in the presence of the sequential time and abstract state postulates.
We republish the relevant part (Part~1) of the technical report here (with slight modifications) to make it easier to access.

\section{The new bounded exploration postulate}

If $f$ is a function symbol in $\U$, $X$ is a state of $A$, $t$ is an $\U$ term, and $T$ is a set of \U\ terms,
then $f_X$ is the interpretation of $f$ in $X$, $\V_X(t)$ is the value of $t$ in $X$, and $\V_X(T) = \set{\V_X(t)\ :\ t\in T}$.

\begin{definition}\rm
States $X,Y$ are \emph{$T$-similar} if
\[
 \V_X(s) = \V_X(t) \iff \V_Y(s) = \V_Y(t)\quad
 \text{for all }s,t\in T. \qefhere
\]
\end{definition}

If states $X,Y$ are $T$-similar then
\begin{equation}\label{sf}
 \s(\V_X(t)) = \V_Y(t)
\end{equation}
is a bijection, the \emph{similarity function}, from $\V_X(T)$ to $\V_Y(T)$. \qef

\begin{lemma}
Suppose that $T$ is closed under subterms, $X$ and $Y$ are $T$-similar, and let term $t = f(t_1,\ldots,t_j)\in T$ and $x_i = \V_X(t_i)$ for $i=1,\dots,j$. Then
\begin{equation}\label{iso}
 \s(f_X(x_1,\dots,x_j)) = f_Y(\s(x_1),\dots,\s(x_j)).
\end{equation}
\end{lemma}

\begin{proof}
\begin{multline*}
 \s(f_X(x_1,\dots,x_j)) =  \s(f_X(\V_X(t_1),\ldots,\V_X(t_j))) = \s(\V_X(t)) =\\
 \V_Y(t) =  f_Y(\V_Y(t_1),\ldots,\V_Y(t_j))  = f_Y (\s(x_1),\ldots, \s(x_j)). \qedhere
\end{multline*}
\end{proof}

\begin{remark}\rm
One may think that, under the hypotheses of Lemma~\ref{iso}, $\s$ is a partial isomorphism from $X$ to $Y$, so that \eqref{iso} holds
whenever $x_1,\dots,x_j$ and $f_X(x_1,\ldots,x_j)$ are in the domain of $\s$.
But this is not necessarily true.  For example, let $a,b$ be nonlogical nullary symbols and $f$ a unary functional symbol in \U.  Set $T = \{a,b\}$ and consider states $X$
and $Y$ with nonlogical elements $1,2,3$ where $f_X(1) =
f_Y(1) = 2$, $f_X(2) = f_Y(2) = 3$, $f_X(3)= f_Y(3) = 1$, and
\[ a_X = a_Y = 1,\quad  b_X = 2,\quad  b_Y = 3. \]
The states $X$ and $Y$ are $T$-similar; in both cases the values of
$a,b$ are distinct.  But $\s$ is not a partial isomorphism because $\s(f_X(1)) = \s(2) = \s(b_X) = b_Y = 3 $ while
$f_Y(\s(1)) = f_Y(\s(a_X)) = f_Y(a_Y) = 2$.\qef
\end{remark}

\begin{definition}\rm
An element $x$ of state $X$ is \emph{$T$-accessible} if
$x = \V_X(t)$ for some $t \in T$.  An update $u = (f,(x_1,\dots,x_j),x_0)$ of $X$ is \emph{$T$-accessible} if all elements $x_i$ are $T$-accessible.
A set of updates of $X$ is \emph{$T$-accessible} if every update in the set is $T$-accessible.

Further, if states $X,Y$ are $T$-similar, $\s: \V_X(T)\to \V_Y(T)$ is the similarity function,  and if $u = (f,(x_1,\dots,x_j),x_0)$ is a $T$-accessible update of $X$, define
\[ \s(u) = \big(f,(\s(x_1),\dots,\s(x_j)),\s(x_0)\big). \qefhere \]
\end{definition}

\begin{postulate}[New Bounded Exploration Postulate]\rm
There exists a finite set $T$ of \U\ terms that is closed
under subterms and such that
\begin{enumerate}
\item[(i)]
$\D(X)$ is $T$-accessible for every state $X$, and
\item[(ii)]
if states $X,Y$ are $T$-similar, $\s: \V_X(T)\to \V_Y(T)$ is the similarity function, and $u$ is an accessible update of $X$, then
\[ u\in \D(X) \iff \s(u) \in \D(Y). \]
In other words, if terms $t_0$ and $f(t_1,\ldots,t_j)$ belong to $T$, $x_i = \V_X(t_i)$, and $y_i = \V_Y(t_i)$, then
\begin{equation*}
(f,(x_1,\ldots,x_j),x_0)\in\Delta(X) \iff
 (f,(y_1,\ldots,y_j),y_0)\in\Delta(Y). \qefhere
\end{equation*}
\end{enumerate}
\end{postulate}

The original bounded exploration postulate did not require the
accessibility of updates.  The accessibility was derived \cite{G141}.

If $T\subseteq T'$ are sets of terms closed under subterms, and if $T$ is a bounded exploration witness, then $T'$ is also a bounded exploration witness.

\begin{example}\rm
We illustrate the necessity of requirement (i).
The vocabulary \U\ of our system $A$ comprises a single nonlogical
function symbol $f$ which is nullary.
Every state $X$ of $A$ has exactly two nonlogical elements, and the element $f_X$ is nonlogical; all states are initial.
Every transition of $A$ changes the value of $f_X$; if $a,b$ are the nonlogical elements of $X$ and $f_X = a$, then
$f_{\tau(X)} = b$.
Clearly, $A$ satisfies the abstract state postulates.

Let $T$ be an arbitrary set of \U\ terms. Then
$T \subseteq \set{\true, \false, \undef, f}$ and $T$ is closed under subterms.
If $X$ is a state with nonlogical elements $a,b$ and $f_X = a$,  then the unique update $(f,b)$ of $X$ is $(f,b)$ which is not $T$-accessible.
Accordingly, $T$ fails requirement~(i).
But, since there are no accessible updates, $T$  satisfies requirement~(ii).

$A$ does not satisfy the original bounded exploration postulate
either.  Indeed, let $X$ be a state with nonlogical elements $a,b$
where $f_X = a$, and let $Y$ be obtained from $X$ by replacing $b$ with a
fresh element $c$.  Then $X$ and $Y$ coincide over every set $T$ of terms but
$\D(X) = \{(f,b)\} \ne \{(f,c)\} = \D(Y)$. \qef
\end{example}

\section{Equivalence of two bounded exploration postulates}

We abbreviate ``bounded exploration'' to BE.

\begin{theorem}
Suppose that $A$ satisfies the sequential state and abstract state
postulates.  Then $A$ satisfies the new BE postulate if and only if it satisfies the original one.
\end{theorem}

\begin{proof}\mbox{}

\medskip\noindent\textbf{Only if.}
We assume that $A$ satisfies the new BE postulate with some BE witness $T$, and we prove that it satisfies the original one with the same BE witness $T$.
Suppose that the states $X$ and $Y$ of $A$ coincide over $T$.
Then $X,Y$ are $T$-similar and the similarity function $\s$ is the identity function from $\V_X(T)$ onto $\V_Y(T)$.  By the new BE postulate, $\D(X) = \D(Y)$.

\medskip\noindent\textbf{If.}
We assume that $A$ satisfies the original BE postulate with a BE witness $T$.
Without loss of generality, $T$ is closed under subterms.
We prove that $A$ satisfies the original BE postulate with the same BE witness $T$.

Statement (i) is proven in \cite[Lemma~6.2]{G141}.
To prove statement~(ii), suppose that $X$ and $Y$ are $T$-similar states of $A$,
$t_0 =f(t_1,\ldots,t_j) \in T$, $x_i = \V_X(t_i)$, and $y_i = \V_Y(t_i)$.
By symmetry, it suffices to prove that $(f,(y_1,\ldots,y_j),y_0) \in
\Delta(Y)$ if $(f,(x_1,\ldots,x_j),x_0) \in \Delta(X)$.  Suppose that
$(f,(x_1,\ldots,x_j),x_0) \in \Delta(X)$.

\medskip\noindent
Case~1: $\V_X(T) \cap \V_Y(T) = \emptyset$.
Create a new state $X'$ from $X$ by replacing $\V_X(t)$ with
$\V_Y(t)$ for every $t\in T$.
States $X'$ and $Y$ coincide over
$T$.  By the old BE postulate, $\D(X') = \D(Y)$.

There is an isomorphism $\xi: X\to X'$ that coincides with the similarity function on $\V_X(T)$ and is
identity otherwise.  $\xi$ naturally lifts to locations,
updates and sets of updates, and we have
\begin{equation*}
(f,(y_1,\ldots,y_j),y_0) = \xi((f,(x_1,\ldots,x_j),x_0)) \in
\xi(\D(X)) = \D(X') = \D(Y).
\end{equation*}

\medskip\noindent
Case~2: $\V_X(T) \cap \V_Y(T) \ne \emptyset$.
Let $\eta$ be an isomorphism from $X$ to a state $X'$ of $A$ such that
$\V_{X'}(T) \cap \V_Y(T) = \emptyset$.
Lifting $\eta$ as above, we have $\V_{X'} (t_i) = \eta x_i$ and
\begin{equation*}
 (f,(\eta x_1,\ldots,\eta x_j),\eta x_0)=\eta ((f,(x_1,\ldots,x_j),x_0)) \in
\eta(\Delta(X)) = \Delta(X').
\end{equation*}
Obviously $X'$ and $Y$ are $T$-similar.  We have Case~1 with $X'$
playing the role of $X$ and $\eta x_i$ playing the role of $x_i$.  Thus $(f,(y_1,\ldots,y_j),y_0)\in\D(Y)$.
\end{proof}

\end{document}